\documentclass[11pt]{article}

\usepackage[margin=1.1in]{geometry}  

\usepackage{graphicx}              
\usepackage{amsmath}               
\usepackage{amsfonts}
\usepackage{dsfont}
\usepackage{amssymb}       
\usepackage{amsthm}                
\usepackage[inline]{enumitem}
\usepackage{color, soul}	
\usepackage{hyperref}
\usepackage[titletoc,title]{appendix}
\usepackage{algorithm}
\usepackage{algpseudocode}
\hypersetup{
	colorlinks=true,
	urlcolor=blue,
	linkcolor=blue,
	citecolor=red
}
\usepackage[table,xcdraw]{xcolor}
\usepackage{caption}
\usepackage{multicol}
\usepackage{tabularx}
\usepackage[
labelfont=sf,
hypcap=false,
format=hang,
width=0.6\columnwidth
]{caption}

\newtheorem{thm}{Theorem}[section]

\newtheorem{prop}[thm]{Proposition}

\newtheorem{defn}[thm]{Definition}

\newtheorem{clm}[thm]{Claim}

\newcommand{\cl}[1]{\mathcal{#1}} 

\usepackage{tikz}

\newcommand\encircle[1]{\tikz[baseline=(char.base)]{
		\node[shape=circle,draw,inner sep=0pt] (char) {#1};}}

\newcommand\blfootnote[1]{%
	\begingroup
	\renewcommand\thefootnote{}\footnote{#1}%
	\addtocounter{footnote}{-1}%
	\endgroup
}

\makeatletter
\renewcommand\@date{{%
		\vspace{-\baselineskip}%
		\large\centering
		\begin{tabular}{@{}c@{}}
			Yakov Babichenko \\
			\small yakovbab@tx.technion.ac.il
		\end{tabular}%
		\quad 
		\begin{tabular}{@{}c@{}}
			Oren Dean \\
			\small orendean@campus.technion.ac.il
		\end{tabular}
		\quad 
		\begin{tabular}{@{}c@{}}
			Moshe Tennenholtz \\
			\small moshet@ie.technion.ac.il
		\end{tabular}
		
		\bigskip
		
		Faculty of Industrial Engineering and Management\par
		Technion --- Israel Institute of Technology\par
		Haifa, Israel
		
		\bigskip
		
		\today
}}

\makeatother

\begin{document}
\title{Paradoxes in Sequential Voting\blfootnote{This project has received funding from the European Research Council (ERC) under the European Union's Horizon 2020 research and innovation programme  (grant agreement n 740435).}}


\maketitle
\begin{abstract}
We analyse strategic, complete information, sequential voting with ordinal preferences over the alternatives. We consider several voting mechanisms: plurality voting and approval voting with deterministic or uniform tie-breaking rules. We show that strategic voting in these voting systems may lead to a very undesirable outcome: Condorcet-winner alternative might be rejected, Condorcet-loser alternative might be elected, and Pareto dominated alternative might be elected. These undesirable phenomena occur already with four alternatives and a small number of voters. For the case of three alternatives we present positive and negative results.	
 \end{abstract}

\section{Introduction}
Traditionally, most game-theoretic models of voting study voting in a simultaneous setting. The reason is, no doubt, due to the fact that we perceive confidentiality as a necessary condition for the fairness of the election process. However, we do see many issues decided upon by means of \emph{sequential} voting. In small committees and even in parliaments, sequential, open ballot is often the default method to make a decision. \\
When we speak of a possible voting outcome we mean the result when the voting is in an equilibrium. However, while simultaneous voting translates to a \emph{normal-form game} and its solution relies on the notion of Nash equilibrium, the sequential voting translates to an \emph{extensive-form game} for which there is the stronger notion of a subgame perfect Nash equilibrium (SPE). This equilibrium is guaranteed to always exist, and when the voters' preferences are totally ordered, it also leads to a unique outcome. \\
While on the one hand a lot of work has been done in axiomatic comparison of simultaneous voting systems (e.g., \cite{FM14}) and on the other hand there has been interest in sequential voting with strategic voters (for instance, \cite{DE10,CX}), the understanding of sequential voting, especially in the case of more than two alternatives, is very limited. As a first step toward the understanding of sequential voting, we analyse it in a complete information setting. Though in voting interactions with many voters this assumption is unrealistic, in small committees it is more plausible.\\
The goal of this paper is to understand how good the performance of sequential voting in scenarios with a small number of voters and a few alternatives (but not necessarily two) is.\\
We consider a setting of voters who have ordinal preferences over the alternatives. In such a setting it is not clear how to measure the quality of the elected alternative. However, there are some scenarios where \begin{enumerate*}[label=\itshape\alph*\upshape)]
	\item it is clear which alternative is desirable, as for instance, if there is a Condorcet winner (Definition~\ref{dfn: Condorcet winner});
	\item it is clear which alternative is undesirable, as for instance, a Condorcet loser (Definition~\ref{dfn: pw loser}) or a Pareto dominated alternative (Definitions~\ref{dfn: weakly dominated} and~\ref{dfn: strongly dominated}).
\end{enumerate*}
We pose the questions: does sequential voting guarantee a choice of a desirable alternative whenever it exists? Does sequential voting guarantee that an undesirable alternative will not be elected?\\
Our results are mainly negative. Although in a two-alternatives setting sequential voting works perfectly well, already with four alternatives the choice of (rejection of) a desirable (undesirable) alternative is not guaranteed. In case of three alternatives, we have both positive and negative results.

\subsection{Related work}
The field of \emph{social choice} deals with mechanisms to collect agents' preferences about different alternatives, and aggregate them to a single ranking (see the classical works~\cite{Arrow1951-ARRIVA} and~\cite{farquharson1969theory} as well as the more recent~\cite{QJPS-13006}). Among the vast branches of this field is voting theory, which deals with axioms for `good' voting rules (e.g., \cite{BCF14, Trust}) and the analysis of voting mechanisms, scoring rules and the like (e.g., \cite{FM14, LP15}). In this paper we restrict our attention to voting mechanisms in which the voters vote sequentially.\footnote{Notice that the term `sequential voting' is often used with respect to a completely different setting, where the voting procedure is divided into several stages, at each stage the voting is for a subset of the alternatives; see, for example,~\cite{CLX09, LXY07}.}\\
Desmedt and Elkind~\cite{DE10} analyse simultaneous and sequential voting systems employing the plurality voting rule. In their model a voter's preference is a total ordering of the power set of the alternatives such that it is consistent with a total ordering of the alternatives themselves; thus their model accepts draws as a legitimate result. They assume the voters are abstain-biased, i.e. a voter will not vote unless pivotal. They prove that in a plurality sequential voting with two alternatives, the winner is the most popular (if one exists), and the voters are truthful whenever they vote. They then show that these nice properties no longer hold when there are three alternatives. For instance, a voter may strategically abstain or vote for an alternative which will not be selected.\\
Conitzer and Xia \cite{CX} discuss sequential voting under a wide range of voting rules, classified by their domination-index (defined to be the smallest number such that any coalition of this size can make any alternative a winner). They prove a general necessary criterion for an alternative to be a winner, and show a voting scenario with an arbitrary number of alternatives  such that the winner is Pareto dominated by all alternatives but one.\\
Dekel and Piccione \cite{DP00}, and later Battaglini \cite{B05}, investigate the information cascade in sequential voting with two alternatives, and compare its equilibria to the symmetric equilibria of simultaneous voting.

\subsection{Our contribution}
In this paper we take a deeper look at the strategic ballots of voters in sequential voting systems. We investigate three types of paradoxical outcomes, namely --- Condorcet winner which is not elected, Condorcet loser which is elected and Pareto dominated alternative which is elected. We consider four types of voting systems which are the four combinations of plurality/approval voting with deterministic/uniform tie-breaking. For each pair of paradox and voting system we show a minimal example (with respect to the number of alternatives) for which there is a scenario with this paradox.

\section{Model and Definitions}\label{sec: classical voting}
Denote by $ V=\{v_1,\ldots,v_n\} $ the set of voters and by $ A=\{a_1,\ldots,a_m\} $ the set of alternatives. We assume that for any $ v_i\in V $, his preferences over the alternatives can be expressed as a total ordering $ (A,\succ_{v_i}) $. We denote by $ \rho=\{\succ_{v_i}\}_{i=1}^{n} $ the collection of the voters' preferences, and call the triplet $ (V,A,\rho) $ a \emph{voting profile}. We assume the voters have complete information. The voters cast their ballots sequentially and publicly (i.e., each voter knows the ballots of those who precede him).
We consider two types of voting rules:
\begin{enumerate}[label=\arabic*.), wide, labelwidth=!, labelindent=0pt]
	\item \emph{Plurality voting:} each voter, in his turn, votes to at most one alternative (abstentions allowed). 
	\item \emph{Approval voting:} each voter may vote to any subset of the alternatives.
\end{enumerate}
In order for a voting system to be well defined we must also define a tie-breaking rule. We consider two tie-breaking rules:
\begin{enumerate}[label=\arabic*.), wide, labelwidth=!, labelindent=0pt]
	\item \emph{Deterministic:} there is a predefined total ordering of the alternatives. The elected alternative is the one that is highest ranked among all alternatives that received the highest number of votes.
	\item \emph{Uniform:} randomly and uniformly pick a winner from all those who received the highest number of votes. In this case we assume that each agent tries to maximize the probability of election of his most favourite alternative and then that of his second-most favourite and so on. Equivalently, we may think of the result as a `winning set' instead of a winning alternative, and ``lift'' the relations of $ \rho $ to be relations on the power set of $ A $ in the following manner. For any $ S\subseteq A $ let $ T_{v_i}(S) $ be the most favourite alternative of $ v_i $ in $ S $. Then for two non-empty subsets $ S_k,S_\ell $,	\begin{flalign*}
	S_k&\succ_{v_i}S_\ell \iff \\
	&[T_{v_i}(S_k)\succ_{v_i}T_{v_i}(S_\ell)] \text{ or } \\
	&[T_{v_i}(S_k)=T_{v_i}(S_\ell) \text{ and } |S_k|<|S_\ell|]\text{ or }\\
	& [T_{v_i}(S_k)=T_{v_i}(S_\ell) \text{ and } |S_k|=|S_\ell| 
	\\&\qquad\text{ and } (S_k-T_{v_i}(S_k))\succ_{v_i}(S_\ell-T_{v_i}(S_k))].
	\end{flalign*}
It's not hard to verify that this relation defines a total ordering.
\end{enumerate}
Thus, we have four combinations of voting rules and tie-breaking rules, which define four different \emph{voting systems}. Game-theoretically, for a given voting profile and a voting system, we have a multi-stage game, describable as an \emph{extensive-form game} --- a tree with all possible voting-sequences, and an outcome at every leaf. The standard solution for this kind of game is a \emph{subgame perfect equilibrium} (SPE). To find an SPE, we start with the last voter. For every voting history, we assume this voter will choose a ballot which gives him the best outcome.\footnote{Although different ballots might lead to the same outcome, the important thing here is the existence and uniqueness of the best outcome.} Moving to the next-to-last voter, we know, for every voting history, how the last voter will respond to any of his ballots. Thus, we can find his best possible outcome for any sequence of voting history. We can continue this backward process until we find the unique outcome that can be achieved in a voting sequence in which every voter selects a best ballot. \\

We now present the paradoxes against which we examine our voting systems.
Let $ (A,\succ_{pw}) $ be the binary relation defined by the pairwise comparison of alternatives. That is, for any $ a_i,a_j\in A $
\begin{flalign*}
 &a_k\succ_{pw} a_\ell \iff \\
 &|\{v_i\in V: a_k\succ_{v_i}a_\ell \}|>|\{v_i\in V: a_\ell\succ_{v_i}a_k \}|. 
\end{flalign*}
An alternative $ a_k $ is called a \emph{Condorcet winner} of $ (V,A,\rho) $ if it beats pairwise all other alternatives, i.e. if $ a_k\succ_{pw}a_\ell $, $ \forall \ell\neq k $. 
\begin{defn}\label{dfn: Condorcet winner}
	A \emph{Condorcet winner paradox} is a voting scenario in which a Condorcet winner of $ (V,A,\rho) $ is not part of the winning set.
\end{defn} 
An alternative $ a_k $ is called a \emph{Condorcet loser} of $ (V,A,\rho) $ if it loses pairwise to all other alternatives, i.e. if $ a_\ell\succ_{pw}a_k $, $ \forall \ell\neq k $. 
\begin{defn}\label{dfn: pw loser}
	A \emph{Condorcet loser paradox} is a voting scenario in which a Condorcet loser of $ (V,A,\rho) $ is the lone winner.
\end{defn} 
An alternative $ a_k $ is \emph{Pareto dominated} by alternative $ a_\ell $ in $ (V,A,\rho) $ if for all $ v_i\in V $, $ a_\ell\succ_{v_i}a_k $. 
\begin{defn}\label{dfn: weakly dominated}
	A \emph{Pareto-dominated weak paradox} is a voting scenario in which a Pareto dominated alternative in $ (V,A,\rho) $ is part of the winning set.
\end{defn} 
\begin{defn}\label{dfn: strongly dominated}
	A \emph{Pareto-dominated strong paradox} is a voting scenario in which a Pareto dominated alternative in $ (V,A,\rho) $ is the lone winner.
\end{defn} 
When the tie-breaking is deterministic the winning set is always a single alternative; hence in this case there is no difference between the Pareto-dominated weak paradox and the Pareto-dominated strong paradox.
\section{The Condorcet Winner and the Condorcet Loser Paradoxes}
When there are only two alternatives the Condorcet winner and Condorcet loser paradoxes actually describe the same situation in which an alternative which is the most favourite by a majority of the voters, is not elected. It is not hard to see that this is impossible in the voting systems we consider here (see Corollary 1 of \cite{DE10} for plurality voting). Here we show that already when there are three alternatives the classical voting systems are no longer resistant to those paradoxes. In all the examples, we list the voters' preferences from high to low. We circle the votes in an SPE.
\begin{clm}\label{prp: plurality Condorcet winner and underdog}
	Already with three alternatives, there are examples of the Condorcet winner paradox and the Condorcet loser paradox, in all of our voting systems.
\end{clm}
\begin{proof}
	Consider the sequential plurality voting scenario in Table~\ref{exm: plurality-deterministic Condorcet winner}, which can either be seen as a five-voter uniform tie-breaking scenario or a four-voter and deterministic tie-breaker (in which case the last voter is the tie-breaker).

It is easy to verify that $ A $	is a Condorcet winner in either case. Since Voter~1 gets his most preferred outcome, his ballot is best possible. If Voters~2, 3 try to improve the outcome by voting for $ A $, then Voter~4 will vote for $ B $ and since Voter~5 (or the tie-breaker) prefers $ B $ over $ A $, $ B $ will be elected. Since $ B $ is worse than $ C $ for Voters~2, 3, their vote for $ C $ is best. Voter~4 cannot change the result since $ C $ is most preferred by Voter~5 (tie-breaker).\footnote{Interestingly, $ C $ is Voter 1's most preferred alternative, but to get it elected he must vote for $ B $ --- his least preferred alternative.}\\
The two scenarios shown in Tables~\ref{exm: plurality-deterministic Condorcet loser} and~\ref{exm: plurality-uniform Condorcet loser} demonstrate a situation where a Condorcet loser wins the election under the plurality-deterministic and plurality-uniform voting systems. In both of them $ C $ is a Condorcet loser. After the first three voters vote for $ C $, this alternative must be elected (because $ C $ is in top priority of the tie-break in case of deterministic tie-breaking; and in case of uniform tie-breaking, it is enough that $ C $ is in top priority of another voter.). Since $ C $ is in top priority of Voter 2, it is enough to show that voters 1 and 3 cannot get a better results. \\
In the deterministic scenario, any other vote of Voter 3 will force Voter 6 to join voters 4 and 5 in their vote for $ A $, since otherwise alternative $ C $ will be elected by tie-breaking. Voter~1 cannot get a better result since Voter~2 can vote for $ B $ which is a best outcome for voters 3 and 6 and is preferred by the tie-break over $ A $. \\
In the uniform scenario, if Voter~3 votes for $ B $, both voters 4 and 5 will vote for $ A $ and Voter~7 will be forced to join them to get the outcome $ \{A,C\} $, which is worse for Voter~3 than just $ C $. If Voter~1 tries to get $ A $ elected, Voter~2 will vote for $ B $ and he will be joined by voters 3, 6 and 7 to get $ B $ as the lone winner, which is the worst outcome for Voter~1.


Below we analyse the SPEs in the remaining four examples.\\

\emph{A Condorcet winner paradox  in approval-deterministic voting (Table~\ref{exm: approval-deterministic Condorcet winner loses})}. Clearly, any ballot of Voter~4 leads to the election of $ C $. Voters~2, 3 cannot get $ A $ elected because $ B $ will have the same number of votes as $ A $, and the tie-breaking gives preference to $ B $. Voter~1 gets his best outcome, hence the voting is in equilibrium.\\

\emph{ A Condorcet winner paradox in approval-uniform voting (Table~\ref{exm: approval-uniform Condorcet winner loses})}. Any vote of Voter~5 gives the same outcome. Voters~3, 4 cannot get $ A $ in the winning set because $ B $ will have more votes than $ A $. Voters~1, 2 get their best result and hence the voting is in equilibrium.\\

\emph{A Condorcet loser paradox in approval-deterministic voting (Table~\ref{exm: approval-deterministic Condorcet loser wins})}. Before the vote of Voter~5, $ C $ has already four voters, thus Voters~5, 6 cannot get a better outcome, even if they co-operate. If Voters~3, 4 try to get $ A $ elected by voting for $ A $ and removing one or two of their votes for $ C $, then Voter~5 will vote for $ B $ or for $ \{B,C\} $ so that Voter~6 will be faced with a choice between $ B $ and $ C $.\footnote{The voting vector before the vote of Voter~6 will be $ (2,3,3) $ and $ C $ is the most preferred alternative in the tie-breaking order.} Voter~6 prefers $ B $ over $ C $, but $ B $ is the worst outcome for Voters~3, 4, thus their vote for $ C $ is their best choice. Voters~1, 2 cannot get $ B $ elected, because Voters~3, 4 and 6 can vote for their most preferred alternative, $ A $, which is also preferred over $ B $ in the tie-breaking order.\\

\emph{A Condorcet loser paradox in approval-uniform voting (Table~\ref{exm: approval-uniform Condorcet loser wins})}. Again, the last voter has no ballot which prevents the worst outcome for him. Voters~5, 6 cannot get $ A $ in the winning set, since alternative $ B $ will have more votes anyway. Voter~4 gets his most preferred outcome. Voter~3 has no good ballot; if, for example, he votes for $ B $, then Voter~4 will simply remove his vote for $ B $ and the rest of the ballots will remain the same. If Voters~1, 2 try to get $ B $ in the winning set by removing their votes for $ C $, alternative $ A $ will be the lone winner with the votes of Voters~3-6; since this is the worst outcome for Voters~1, 2, they are better off with their current vote.\\

	\begin{minipage}[h]{\columnwidth}
		\centering
		\begin{tabular}{rccc}
			Voter 1 & C  & A & \encircle{B} \\
			Voter 2 & A & \encircle{C} & B \\
			Voter 3 & A & \encircle{C} & B \\
			Voter 4 & \encircle{B} & A & C \\
			Voter 5/Tie-Breaker & \encircle{C} & B & A
		\end{tabular}		
		\captionof{table}{Plurality voting. $ A $ is a Condorcet winner, $ C $ is elected.}
		\label{exm: plurality-deterministic Condorcet winner}		
\end{minipage}
\noindent\rule{\columnwidth}{0.4pt}\vspace{2mm}
	\begin{minipage}[h]{\columnwidth}
		
		\begin{multicols}{2}
			\centering
			\begin{tabular}{rccc}
				Voter 1 & A & \encircle{C} & B \\
				Voter 2 & \encircle{C} & B & A \\
				Voter 3 & B & \encircle{C} & A \\
				Voter 4 & \encircle{A} & B & C \\
				Voter 5 & \encircle{A} & B & C \\
				Voter 6 &  \encircle{B} & A & C \\
				Tie-Break & \encircle{C} & B & A
			\end{tabular}
			\captionof{table}{Plurality-deterministic voting. $C $ is a Condorcet loser and gets elected.}
			\label{exm: plurality-deterministic Condorcet loser}		
			\columnbreak			
			\centering
			\begin{tabular}{rccc}
				Voter 1 & A & \encircle{C} & B \\
				Voter 2 & \encircle{C} & B & A \\
				Voter 3 & B & \encircle{C} & A \\
				Voter 4 & \encircle{A} & B & C \\
				Voter 5 & \encircle{A} & B & C \\
				Voter 6 & \encircle{C} & B & A\\			
				Voter 7 &  \encircle{B} & A & C 
			\end{tabular}	
			\captionof{table}{Plurality-uniform voting. $C $ is a Condorcet loser and the lone winner.}
			\label{exm: plurality-uniform Condorcet loser}		
		\end{multicols}	
	\end{minipage}
\noindent\rule{\columnwidth}{0.4pt}\vspace{2mm}
\begin{minipage}[h]{\columnwidth}
\begin{multicols}{2}
	
	\centering
	\begin{tabular}{rccc}
		Voter 1 & \encircle{C} & A &\encircle{B}\\
		Voter 2 & A & \encircle{C} & B \\
		Voter 3 & A & \encircle{C} & B \\		
		Voter 4 & \encircle{B} & A & C \\			
		Tie-Break & B & \encircle{C} & A
	\end{tabular}
	\captionof{table}{Approval-deterministic voting. $A $ is a Condorcet winner, $ C $ is elected.}
	\label{exm: approval-deterministic Condorcet winner loses}	
	\columnbreak
	\centering	
			\begin{tabular}{rccc}
				Voter 1 & \encircle{C} & A & \encircle{B} \\
				Voter 2 & \encircle{C} & A & \encircle{B} \\
				Voter 3 & A & \encircle{C} & B \\
				Voter 4 & A & \encircle{C} &  B\\
				Voter 5 & \encircle{B} & A &  C	
			\end{tabular}
		\captionof{table}{Approval-uniform voting. $A $ is a Condorcet winner, $ C $ is elected.}
		\label{exm: approval-uniform Condorcet winner loses}			
			
\end{multicols}	
\end{minipage}
\noindent\rule{\columnwidth}{0.4pt}\vspace{2mm}
\begin{minipage}[h]{\columnwidth}
	\begin{multicols}{2}
			\centering
		\begin{tabular}{rccc}
			Voter 1 & \encircle{B} & \encircle{C} & A\\
			Voter 2 & \encircle{B} & \encircle{C} & A\\
			Voter 3 & A & \encircle{C} & B \\
			Voter 4 & A & \encircle{C} & B \\		
			Voter 5 & \encircle{B} & A & C \\			
			Voter 6 & \encircle{A} & B & C \\			
			Tie-Break & \encircle{C} & A & B
		\end{tabular}
		\captionof{table}{Approval-deterministic voting. $C $ is a Condorcet loser gets elected.}
		\label{exm: approval-deterministic Condorcet loser wins}			
	\columnbreak	
	
	\begin{tabular}{rccc}
		Voter 1  & \encircle{B}& \encircle{C} & A \\
		Voter 2  & \encircle{B}& \encircle{C} & A \\
		Voter 3 & \encircle{A}  & B & C \\
		Voter 4 & \encircle{C} & A & \encircle{B} \\
		Voter 5 & A & \encircle{C} & B \\
		Voter 6 & A & \encircle{C} & B \\
		Voter 7 & \encircle{B} & A & C
	\end{tabular}
	\captionof{table}{Approval-uniform voting. $C $ is a Condorcet loser and the lone winner.}
	\label{exm: approval-uniform Condorcet loser wins}				
\end{multicols}	
\end{minipage}
\noindent\rule{\columnwidth}{0.4pt}\vspace{2mm}\\
\end{proof}

\section{Pareto-efficiency}
The following Pareto-dominated paradox scenario for the plurality-deterministic voting system was demonstrated in~\cite{CX}. Suppose we have one voter with preference order $ A\succ C\succ B $ and another voter with preference $ B\succ A\succ C $. Furthermore, assume the tie-breaking order is $ C\succ B\succ A $. The first voter should not vote for $ A $ as that would lead to the election of $ B $ which is his worst outcome. He therefore votes for $ C $, which is elected. Notice that  $ A $ Pareto dominates $ C $ in the voters' preferences, but the tie-breaking order gives preference to $ C $ over $ A $. We claim that the plurality-deterministic voting system is the only one of our four voting systems which allows this paradox with only three alternatives.

\begin{prop}\label{prp: no dominated with 3 alt}
	Let $ \cl{S} $ be a voting system with three alternatives and at least two voters. If $ \cl{S} $ is approval voting with either tie-breaking rules, or it is plurality voting with uniform tie-breaking, then the Pareto-dominated weak paradox cannot happen.
\end{prop}
\begin{proof}
	Assume that we have a voting scenario with alternatives $ \{A,B,C\} $ and $ C $ is Pareto dominated by $ A $. For $ X\in\{A,B,C\} $, we call a voter which has $ X $ at the top of his preference order an `$ X $-type' voter. Since $ C $ is Pareto dominated, there are only $ A $-type and $ B $-type voters. Now, in any voting system, if there is a strict majority of $ X $-type voters then they all just vote for $ X $ in an SPE (since then $ X $ is elected alone and that's the best possible result for them). We therefore assume that there is exactly the same number of $ A $-type and $ B $-type voters. If the tie-breaking is uniform then when all voters vote for their top alternative we get an SPE with the result $ \{A,B\} $, and no voter can get a better result. The remaining case is when $ \cl{S} $ is approval-deterministic. Assume for contradiction that $ C $ is selected in an SPE. Since all the $ B $-type voters rank $ C $ last, any other outcome is better for them. This means that every votes of the $ B $-type voters leads to the election of $ C $. We may assume then, that all the $ B $-type voters vote just for $ B $. Since half of the voters are $ B $-type, it must be that the tie-breaking prefers $ C $ over $ B $. Now assume that all the $ B $-type voters vote just for $ A $. It must be that no $ A $-type voter votes for $ A $, otherwise $ A $ would have got more than half of the votes and get elected. Moreover, since there the same number of $ A $-type and $ B $-type voters, it must be that all the $ A $-type voters vote for $ C $. Let $ v_i $ be the last $ A $-type voter. Assume that $ v_i $ adds $ A $ to his ballot and votes for $ \{A,C\} $. This cannot be a best vote, as explained above. The outcome after this vote must be $ B $, since this is the only outcome which is worse than $ C $ for this voter. But this is not possible since we already saw that the tie-breaking order prefers $ C $ over $ B $ and there cannot be more votes for $ B $ than for $ C $ (since we assume that all the $ A $-type voters already voted for $ C $). We have reached a contradiction, which means that $ C $ does not get elected in an SPE.
\end{proof}

As mentioned above, there is an example of a Pareto-dominated paradox with three alternatives when the voting system is plurality-deterministic. The next claim, then, settles the question of Pareto-efficiency for the systems with deterministic tie-breaking.

\begin{clm}\label{prp: dominate-paradox}{\label{prp: deterministic dominated}}
	If there are at least four alternatives, then there are examples of the Pareto-dominated paradox for the approval voting system with deterministic tie-breaking.
\end{clm}
\begin{proof}
In the scenario of Table~\ref{exm: approval-deterministic pareto dominated} $ C $ is Pareto dominated by $ A $. The first two voters vote for $ \{C,B\} $ thus creating a threat on Voters 3 and 4 that if they don't vote for $ C $ as well, then $ B $ will be elected. If Voters~1 and~2 change their vote (for example, for $ \{A,B\} $), then Voter 3 will vote for $ \{D,A\} $, Voter 4 will vote for $ D $ and Voter 5 will be forced to vote for $ D $ as well, making $ D $ the winner.\\

\begin{minipage}[h]{\columnwidth}
	\centering
		\begin{tabular}{l@{\hskip 0.1in}l@{\hskip 0.1in}l@{\hskip 0.1in}l@{\hskip 0.1in}l}
	Voter 1     & A & \encircle{C} & \encircle{B} & D \\
	Voter 2     & A & \encircle{C} & \encircle{B} & D \\
	Voter 3     & D & A & \encircle{C} & B \\
	Voter 4     & D & A & \encircle{C} & B \\
	Voter 5     & \encircle{B} & D & A & C \\
	Tie-Break & D & A & B & \encircle{C}
\end{tabular}
\captionof{table}{Approval-deterministic voting. $C $ is Pareto dominated by $ A $ and gets elected.}
\label{exm: approval-deterministic pareto dominated}		
\end{minipage}
\noindent\rule{\columnwidth}{0.4pt}\vspace{2mm}\\
\end{proof}

\begin{clm}\label{prp: approval dominated}
	If there are at least four alternatives, then there are examples of the Pareto-dominated weak paradox for both plurality and approval voting systems with uniform tie-breaking.
\end{clm}
\begin{proof}
	We start with the plurality-uniform voting (Table~\ref{exm: plurality-uniform pareto dominated}). Voters~5, 6 cannot get a better outcome, even if they co-operate. Voter~4 must vote for $ C $ in order to force Voters~5, 6 to vote for $ A $.\footnote{A vote for $ A $ will lead to the outcome $ \{A,B,C\} $ which is worse for him.} If either Voter~2 or Voter~3 change their ballots, Voter~4 will not be able to force Voters~5, 6 to vote for $ A $; instead, Voters~4, 5 and~6 will vote for $ B $ which will lead to a worse outcome for Voters~2, 3. If Voter~1 change his voter, for example for $ A $, then the rest of the voters will all vote for their top alternative, which would lead to the outcome $ \{A,B,D\} $, which is worse for Voter~1.\\
	
	In the case of approval-uniform voting (Table~\ref{exm: approval-uniform pareto dominated}). Voter~5 votes for both $ A $ and $ C $ so that Voter~6 is not able to get $ B $ in the winning set, and thus is forced to vote for $ A $. Voter~4 has no vote which changes the outcome in his favour. Voters~2, 3 cannot get $ D $ in the winning set because of the threat of Voters~4, 5 and~6 to vote for $ B $. Voter~1 must vote for $ C $ and allow it to be part of the winning set, since this is what allows Voter~5 to force Voter~6 to vote for $ A $ as explained above.\\

\noindent
\begin{minipage}[h]{\columnwidth}	
	\begin{multicols}{2}
		\centering
		\begin{tabular}{lcccc}
			Voter 1 & A & B & D & \encircle{C} \\
			Voter 2 & D &  \encircle{A} & C & B \\
			Voter 3 & D & \encircle{C} & B & A \\
			Voter 4 & A & B & D & \encircle{C} \\
			Voter 5 & B & \encircle{A} & D & C \\
			Voter 6 & B & D & \encircle{A} & C
		\end{tabular}
	\captionof{table}{Plurality-uniform voting. $C $ is Pareto dominated by $ D $. The winning set is $ \{A,C\} $.}
	\label{exm: plurality-uniform pareto dominated}		
		\columnbreak
		
		\begin{tabular}{lcccc}
			Voter 1 &  \encircle{A} & \encircle{B} & D & \encircle{C} \\
			Voter 2 &  D  & \encircle{C} & B & A \\
			Voter 3 & D & \encircle{C} &  \encircle{A} & B \\
			Voter 4 &\encircle{B} & D & C & A \\
			Voter 5 &   \encircle{A}  & B & D & \encircle{C} \\
			Voter 6 &  B &  D  &  \encircle{A} & C
		\end{tabular}
	\captionof{table}{Approval-uniform voting. $C $ is Pareto dominated by $ D $. The winning set is $ \{A,C\} $.}
	\label{exm: approval-uniform pareto dominated}				
	\end{multicols}		
\end{minipage}
\noindent\rule{\columnwidth}{0.4pt}\vspace{2mm}\\
\end{proof}
The question of whether systems with uniform tie-breaking are weakly Pareto-efficient remains open.

\bibliographystyle{splncs03}
\bibliography{paper}

\end{document}